%% file: main.tex
\documentclass[showpacs,preprintnumbers,amsmath,amssymb,12pt]{article}
\usepackage[utf8]{inputenc}
\usepackage[italian,english]{babel}
\usepackage{amsmath}
\usepackage{graphicx}
\usepackage[colorinlistoftodos]{todonotes}
\usepackage{caption}
\usepackage{subcaption}
\usepackage{tabularx}
\usepackage{bbm}
\usepackage{listings}
\usepackage{amsthm}
\usepackage{algorithm}
\usepackage{algpseudocode}
\usepackage{mathtools}
\usepackage{amsmath}
\usepackage[toc,page]{appendix}
\usepackage{booktabs}
\usepackage{xcolor}
\usepackage{amsfonts}
\usepackage{epstopdf}
\usepackage{adjustbox}
\usepackage{csquotes}

\usepackage[numbers]{natbib}

\makeatletter
\def\BState{\State\hskip-\ALG@thistlm}
\makeatother

\usepackage{amsmath,amsfonts,amssymb,amsthm,epsfig,epstopdf,titling,url,array}

\theoremstyle{plain}
\newtheorem{thm}{Theorem}[section]

\newtheorem{prop}[thm]{Proposition}

\theoremstyle{plain}

\theoremstyle{remark}

\newcommand{\pdf}{\textit{pdf}}
\newcommand{\chf}{\textit{chf}}
\newcommand{\cdf}{\textit{cdf}}

\newcommand{\rv}{\textit{rv}}
\newcommand{\id}{\textit{id}}

\newcommand{\sd}{\textit{sd}}

\newcommand{\eqd}{\stackrel{d}{=}}

\newcommand{\arem}{$a$-remainder}

\newcommand{\BM}{\emph{BM}}

\newcommand{\Levy}{L\'{e}vy}

\title{A bivariate Normal Inverse Gaussian process with stochastic delay: efficient simulations and applications to energy markets}

\author{
	Matteo Gardini\thanks{Department of Mathematics, University of Genoa, Via Dodecaneso 16146, Genoa, Italy, email gardini@dima.unige.it} 
	\and 
	Piergiacomo Sabino\thanks{Quantitative Modelling E.ON SE
		Br\"usseler Platz 1, 45131 Essen, Germany, email piergiacomo.sabino@eon.com} 
	\and
	 Emanuela Sasso\thanks{Department of Mathematics, University of Genoa, Via Dodecaneso 16146, Genoa, Italy, email sasso@dima.unige.it}}

\date{\today}
\begin{document}

\maketitle


\begin{abstract}
Using the concept of self-decomposable subordinators introduced in \citet{gardini2020}, we build a new bivariate Normal Inverse Gaussian process that can capture stochastic delays. In addition, we also develop a novel path simulation scheme that relies on the mathematical connection between self-decomposable Inverse Gaussian laws and \Levy-driven Ornstein-Uhlenbeck processes with  Inverse Gaussian stationary distribution. We show that our approach provides 
an improvement to the existing simulation scheme detailed in \citet{Zhang2008} because it does not rely on an acceptance-rejection method.  

Eventually, these results are applied to the modelling of energy markets and to the pricing of spread options using the proposed Monte Carlo scheme and Fourier techniques.
\end{abstract}

\section{Introduction and preliminaries}
\label{sec:Intro}
\input{Introduction.tex}

\section{Self decomposable NIG process}
\label{sec:SDMODELS}

\input{SelfDecomposableNIGProcess.tex}

\section{Simulation Algorithm}
\label{sec:SimAlg}
\input{SimulationAlgorithm.tex}
\section{Financial Application}
\label{sec:NumRes}
\input{NumericalResults.tex}
\section{Conclusions}
\label{sec:Conclusions}
\input{Conclusion.tex}
\clearpage
\appendix
\section{IG laws parametrization}
\label{sec:Appendix}
\input{differentparametrizations.tex}

\clearpage
\bibliographystyle{plainnat}
\bibliography{biblioAll}

\end{document}

%% file: Introduction.tex

The \citet{BLS1973} model is probably the most popular stochastic model used to describe the dynamics of financial indices. Even though it is well-known that it is not able to capture many  stylized facts, its simplicity and its flexibility often make it the standard choice for many financial applications. In the univariate setting several models have been proposed to overcome its limits, relying, for example, on more general \Levy\ processes. However, in a multi-market setting, the Black-Scholes model is still a milestone due to the fact that alternatives can be less mathematically tractable and their calibration can be computationally demanding. An attempt to combine tractability and simple calibration in a multivariate framework has been proposed by \citet{Semeraro2008} in the context of Variance Gamma (VG) processes, introduced by \citet{MadanSeneta90}  and extended to Normal Inverse Gaussian (NIG) processes, introduced by \citet{BN98}, in \citet{SL2010}. Each marginal of the multivariate process is built via Brownian subordination: the resulting subordinator is the sum of an independent subordinator and a  subordinator shared by all the components, both mutually independent. Therefore, the construction has a nice financial interpretation  in terms of idiosyncratic and systematic risks. An alternative approach to construct multidimensional \Levy\ processes including VG and NIG processes, has been proposed by \citet{BB2013} sharing the same logic of idiosyncratic and systematic risks.
Such models are able to capture some empirical facts such as discontinuities in price trajectories, volatility smiles and non-normality in log-returns, whereas the joint dependence is driven by the common systematic component.

On the other hand, the impact of new information in one market might require some time to be propagated onto dependent markets therefore, the aforementioned models cannot replicate any \textit{stochastic delay} or any \textit{synaptic risk} as named in \citet{NCPPS2018}. Indeed, it is not so rare to observe that the impact on other related markets occurs after a stochastic time delay. Nowadays, a clear example is offered by the recent pandemic disease of Covid-19: as one can see in Figure \ref{fig:covidImpact}, first blown cases appeared at the beginning of January 2020 in China leading to a big downward jump in Shangai's index after a flat period due to Chinese New Year celebrations and subsequently the virus spread all over the world. Italy registered first cases at the end of February, Brazil at the beginning of March  leading to a general drop in the whole world economy.

\begin{figure}
    \centering
    \includegraphics[scale=0.3]{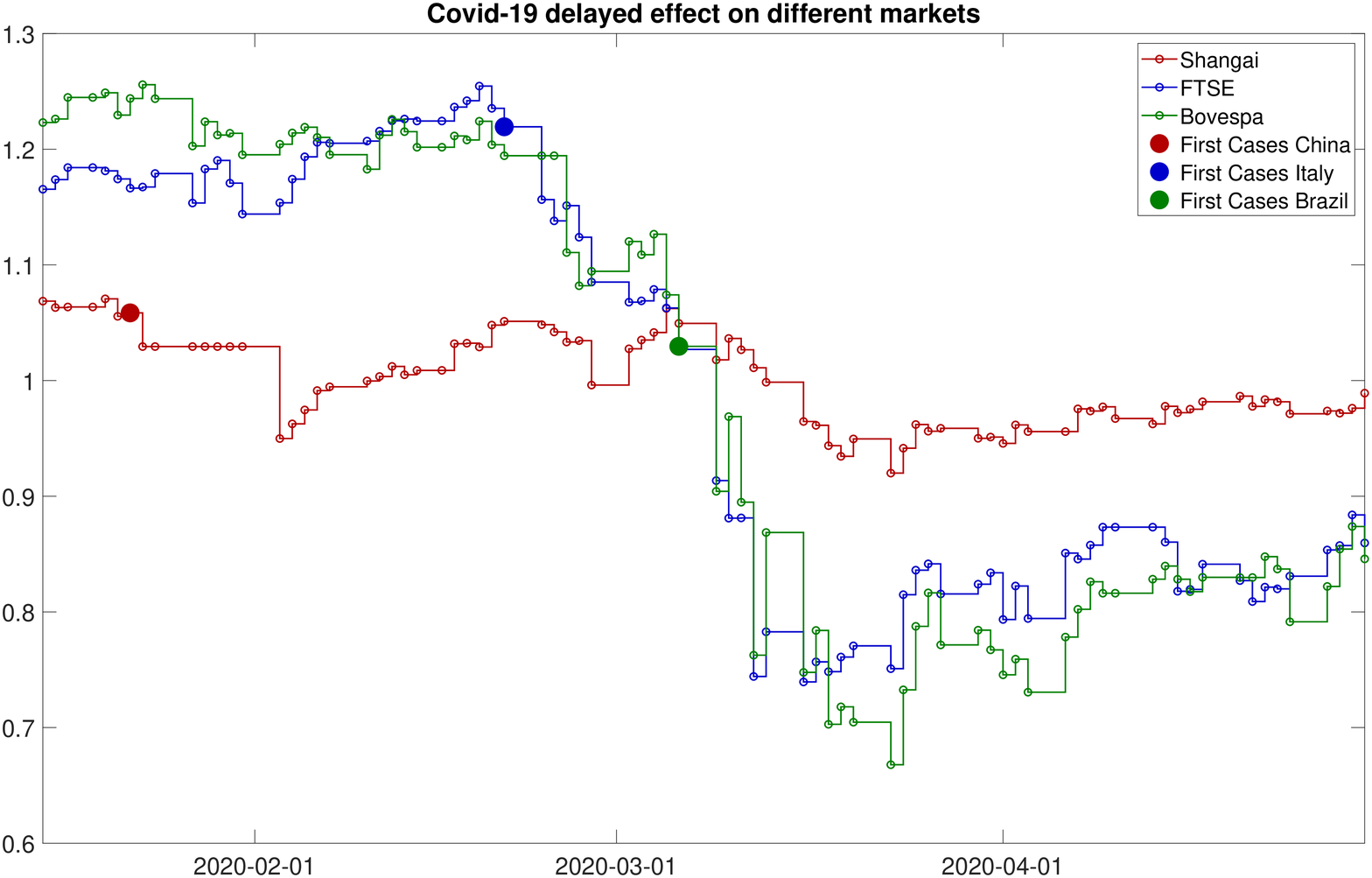}
    \caption{Impact of Covid-19 disease in some markets world wide.}
    \label{fig:covidImpact}
\end{figure}

Recently in \citet{gardini2020}  we have shown how the notion of self-decomposability (\sd) can be used to describe stochastic delays and to introduce synaptic risk in financial models.
We recall that the law of a \rv\ $X$ is said to be \sd\ (see \citet{Sato} and \citet{Cufaro08}) if for every $a \in \left(0,1\right)$ its characteristic function (\chf) $\phi\left(u\right)$ can be represented as 
\begin{equation}
\phi_{X}\left(u\right) = \phi_{X}\left(au\right)\chi_{a}\left(u\right).
\label{eqn:chfsdlaws}
\end{equation}
with $\chi_{a}\left(u\right)$  also a \chf. It means that   
we can always find two independent \rv's $Y$ (with the same law of $X$) and $Z_{a}$ such that, in distribution
\begin{equation}\label{eq:def:sd}
X \stackrel{d}{=} aY + Z_{a}
\nonumber
\end{equation}
where of course, $\chi_a(u)$ is the \chf\ of $Z_a$ (hereafter called the \arem\ of the law of $X$) whose law is infinitely divisible (\id) but not in general \sd\ (see \citet{Sato}).

Based on these facts, in \citet{gardini2020} we have introduced what we name \sd-\textit{subordinators} that are the  building blocks for the construction of correlated \Levy\ processes.
In analogy to \sd\ laws, such subordinators $H_{1}\left(t\right), H_{2}\left(t\right)$ are defined as follows 
\begin{equation}
    H_{2}\left(t\right) = aH_{1}\left(t\right) + Z_{a}\left(t\right)
    \label{eqn:subordinatorsSD}
\end{equation}
where $H_2(t)$ and $Z_a(t)$ are independent processes. The last  equation  is mathematically well-posed and  has a clear  interpretation: the stochastic time processes $H_{1}\left(t\right), H_{2}\left(t\right)$ \enquote{run together} with a stochastic delay $Z_{a}\left(t\right)$ that is controlled by $a$ that simply plays the role of the instantaneous correlation between $H_1(t)$ and $H_2(t)$.  When $a$ tends to $1$ then $H_{1}\left(t\right)$ and $H_{2}\left(t\right)$ become essentially indistinguishable. By subordinating Brownian motions (BM) with such subordinators we  construct a class of dependent processes that are at least marginally \Levy. This means that we can extend the approaches of \citet{Semeraro2008}, \citet{SL2010} and \citet{BB2013} to cover stochastic delays while keeping mathematical tractability, easy calibration and clear financial interpretation.

This study can be considered the sequel of \citet{gardini2020}, where now our main focus is on  bivariate \sd\ Inverse Gaussian (IG) suborndinators and on the construction of bivariate dependent NIG processes. The first contribution of this work is the derivation of closed form formulas for the linear correlation and the \chf\ of the last processes. These results are instrumental for the  calibration and the pricing of derivative contracts.  The pricing of complex derivative contracts is often accomplished via Monte Carlo (MC) simulations. To this end, a second contribution of this study consists of a novel and efficient algorithm to generate the \arem\ of IG laws and therefore to simulate the skeleton of $Z_a(t)$, of the \sd\ IG subordinators and of the bivariate NIG processes. As already observed among others in \citet{Taufer2009}, \citet{Sabino2020b} and \citet{cs20}, the transition law between $t$ and $t + \Delta t$ of a \Levy-driven Ornstein-Uhlenbeck (OU) $X(t)$ having a certain  stationary law coincides with that of the \arem\ of such a law by setting $a=e^{-\lambda \Delta t}$ where $\lambda$ is the mean-reversion rate of $X(t)$. Hence, the simulation of the \arem\ of a IG law is equivalent to the simulation of the skeleton of a IG-OU process, this last one having been illustrated in \citet{Zhang2008}. We show that our proposal is more efficient than that of \citet{Zhang2008}, because it does not rely on acceptance-rejection methods. Note that being $Z_a(t)$ a \Levy\ process, its simulation requires the same $a$ at all times $t$ while instead $a=e^{-\lambda \Delta t}$ depends on the time step $\Delta t$. 

Finally, we illustrate the applicability of the proposed bivariate \sd-NIG processes to energy markets and in particular to the pricing of spread options via MC simulations and Fourier techniques.

\par The article is organized as follow: Section \ref{sec:SDMODELS} introduces \sd-NIG processes and their mathematical properties. In Section \ref{sec:SimAlg} we describe the method to simulate $Z_{a}\left(t\right)$ and hence  $H_{1}\left(t\right),H_{2}\left(t\right)$ defined in Equation \eqref{eqn:subordinatorsSD}. In Section \ref{sec:NumRes} we apply the models described in Section \ref{sec:SDMODELS} to power and gas forward markets and to  the pricing of spread options with MC and Fourier Techniques.  Section \ref{sec:Conclusions} concludes the paper with an overview of future inquires and possible further applications.


%% file: SelfDecomposableNIGProcess.tex
The NIG process is constructed via the subordination of a BM with an IG process. On the other hand, there are different characterizations of the \pdf\ of an IG law: we denote the notation using the parameter-setting $\left(\mu,\lambda\right)$, adopted for instance in \citet{CT2003}, with $IG_{T}\left(\mu,\lambda\right)$: within this setting $\mu>0$ is the mean and $\lambda>0$ is the shape parameter. On the other side, we refer to the original notation in \citet{BN97} with $IG_{B}\left(a,b\right)$: in this case $a >0$ and $b>0$ describe the scale and the shape of the distribution, respectively. In Appendix \ref{sec:Appendix} we give some details on how to switch from one to the other. In general, the $IG_{B}$ notation is convenient to analyze sums of IG \rv's, whereas $IG_{T}$ is more convenient to work with expectations and \chf.

\citet{Semeraro2008}, \citet{SL2010} and \citet{BB2013} proposed a simple technique to introduce dependence between \Levy\  processes: given three \Levy\ independent processes $X_{1}\left(t\right)$, $X_{2}\left(t\right)$ and $Z\left(t\right)$ and $a_{1},a_{2}\in \mathbb{R}$ one can set:
\begin{align*}
Y_{1}\left(t\right) = X_{1}\left(t\right) + a_{1}Z\left(t\right) \\
Y_{2}\left(t\right) = X_{2}\left(t\right) + a_{2}Z\left(t\right) \\
\end{align*}
The processes $Y_{1}\left(t\right)$ and $Y_{2}\left(t\right)$ are clearly dependent, because of the common process $Z\left(t\right)$. This idea can be applied  to different types of processes, included subordinators. The economic interpretation is clear: $Z\left(t\right)$ represents the \textit{systematic risk} whereas $X_{j}\left(t\right)$ models the \textit{idiosyncratic risk}. Of course, this simple construction can be applied to obtain multivariate VG and NIG processes. Nevertheless, as mentioned in the introduction, these settings cannot cover stochastic-delay and what we call \textit{synaptic risk}. In \citet{gardini2020} we detailed on the construction of bivariate \sd-VG processes, whereas in this sequel we focus on the bivariate version of \sd-NIG processes. In a nutshell, our approach consists of replacing the common and marginal-specific subordinators of \citet{Semeraro2008}, \citet{SL2010} and \citet{BB2013} with \sd-subordinators defined in \eqref{eqn:subordinatorsSD}.

\subsection{Semeraro \sd-NIG model}
\label{sec:NIGSemeraro}
In this subsection we illustrate the steps required to extend the model proposed by \citet{Semeraro2008} in order to cope with stochastic delay relying on the \sd\ subordinators of Equation \eqref{eqn:subordinatorsSD}.  

Let $I_{j}\left(t\right) \; j=1,2$ be independent subordinators, and $H_{1}\left(t\right)$,  $H_{2}\left(t\right)$ be \sd\ subordinators defined in \eqref{eqn:subordinatorsSD}, independent of $I_{j}\left(t\right)$. Define the subordinator $G_{j}\left(t\right)$

\begin{equation}
G_{j}\left(t\right) = I_{j}\left(t\right) + \alpha_{j}H_{j}\left(t\right), \quad j=1,2
\label{eqn:SSubordinators}
\end{equation}
with $\alpha_{j} \in \mathbb{R}^{+}$. Let now be $\mu_{j} \in \mathbb{R}$, $\sigma_{j} \in \mathbb{R}^{+}$ and $W_{j}\left(t\right)$ standard independent BM's, we define the subordinated \BM\ $Y_{j}\left(t\right)$ as:
\begin{equation}
Y_{j}\left(t\right) = \mu_{j}G_{j}\left(t\right) + \sigma_{j}W_{j}\left(G_{j}\left(t\right)\right), \quad j=1,2.
\label{eqn:subordinatedprocess}
\end{equation}

We remark that when $a$ in \eqref{eqn:subordinatorsSD} tends to $1$ there is no time delay and the synaptic risk coincides with the systematic risk as in the original approach of \citet{Semeraro2008}.  \\

\par A bivariate NIG process with IG \sd-subordinators can be defined starting from \eqref{eqn:SSubordinators} in the following way. Let be $\alpha_{j} =\gamma_{j}^{2}$ and let $I_{j}\left(t\right)$ and $H_{j}\left(t\right)$ be distributed as follows:
\begin{equation}
\begin{split}
I_{j}\left(t\right) & \sim  IG_{T}\left(\frac{A_{j}\gamma_{j}t}{B},A_{j}^{2}t^{2}\right) \\
H_{j}\left(t\right) &\sim IG_{T}\left(\frac{At}{B},A^{2}t^{2}\right)
\end{split}
\label{eqn:distosubsemeraro}
\end{equation}
and hence we get:
\begin{align*}
    G_{j}\left(t\right) \sim  IG_{T}\left(\frac{\left(A_{j}+A\gamma_{j}\right)\gamma_{j}t}{B},\left(A_{j}+A\gamma_{j}\right)^{2}t^{2}\right).
\end{align*}
Since $G\left(t\right)$ is a stochastic time, it is customary to require that $\mathbb{E}\left[G_j\left(t\right)\right]=t$: this condition can be easily fulfilled by imposing:

\begin{equation*}
A_{j} + A \gamma_{j} = \frac{B}{\gamma_{j}}.
\end{equation*}

\noindent Consequently, denoting with $k_{j}$ the variance of the subordinator $G\left(t\right)$ at time $t=1$, we have that:
\begin{equation*}
k_{j}\coloneqq Var\left[G_{j}\left(1\right)\right] = \frac{1}{\left(A_{j}+A\gamma_{j}\right)^{2}} = \frac{\gamma_{j}^{2}}{B^{2}}.
\end{equation*}
As observed in \citet{SL2010}, assuming $B=1$ is not restrictive: hence $k_{j}=\gamma_{j}^{2}$ and then $k_{j} = \alpha_{j}$. After simple calculations, one can find that the expression of the (instantaneous) linear correlation coefficient at time $t$ of the process $\boldsymbol{Y}\left(t\right) = \left(Y_{1}\left(t\right),Y_{2}\left(t\right)\right)$ is:

\begin{equation}
\rho_{\left(Y_{1}\left(t\right),Y_{2}\left(t\right)\right)} = \frac{\mu_{1}\mu_{2}\alpha_{1}\alpha_{2}a A}{\sqrt{\sigma_{1}^{2} + \mu_{1}^{2}\alpha_{1}}\sqrt{\sigma_{2}^{2} + \mu_{2}^{2}\alpha_{2}}}
\label{eqn:linCorrS}
\end{equation}
Compared to the formula of the linear coefficient in \citet{Semeraro2008} the equation above has an additional parameter $a$ that tunes the stochastic delay.  

Finally, the \chf\ of $\boldsymbol{Y}\left(t\right)$ is given by the following proposition.

\begin{prop}
\label{cor:2DNIGchfSemeraro}
Denote $\phi\left(u;\mu,\lambda\right)$ the \chf\ of a  \rv\ distributed according to a $IG_{T}\left(\mu,\lambda\right)$ law then
the joint \chf\ at time $t$ of $\boldsymbol{Y}\left(t\right)$ of Equation \eqref{eqn:subordinatedprocess}, where $H_{j}\left(t\right)$ and $I_{j}\left(t\right)$ are distributed as in \eqref{eqn:distosubsemeraro} for $j=1,2$, is:
\begin{equation}
\begin{split}
\phi_{\boldsymbol{Y}\left(t\right)}\left(\boldsymbol{u}\right) = &\phi_{I_{1}\left(t\right)}\left(u_{1}\mu_{1} + i \frac{\sigma_{1}^{2} u_{1}^{2}}{2}\right) 
\phi_{I_{2}\left(t\right)}\left(u_{2}\mu_{2} + i \frac{\sigma_{2}^{2} u_{2}^{2}}{2}\right)
\phi_{Z_{a}\left(t\right)}\left(u_{2}\mu_{2} + i \frac{\sigma_{2}^{2} u_{2}^{2}}{2}\right) \\
& \phi_{H_{1}\left(t\right)}\left(\alpha_{1}\left(u_{1}\mu_{1} + i \frac{\sigma_{1}^{2} u_{1}^{2}}{2}\right) + a \alpha_{2}\left(u_{2}\mu_{2} + i \frac{\sigma_{2}^{2} u_{2}^{2}}{2}\right)\right)
\end{split}
\label{eqn:generalchf}
\end{equation}
where 

\begin{equation}
\begin{split}
    \phi_{H_{j}\left(t\right)}\left(u\right) &= \phi\left(u;At,A^{2}t^{2}\right),   \quad j=1,2 \\
    \phi_{I_{j}\left(t\right)}\left(u\right) &= \phi\left(u;A_{j}t\gamma_{j},A_{j}^{2}t^{2}\right),  \quad j=1,2 \\
    \phi_{Z_{a}\left(t\right)}\left(u \right) & = \frac{\phi\left(u;At,A^{2}t^{2}\right)}{\phi\left(au;At,A^{2}t^{2}\right)}
    \label{eqn:chfNigSem}
\end{split}
\end{equation}
\end{prop}
\begin{proof}
The proof follows the scheme we used to prove the Proposition~3.5 of \citet{gardini2020}. 
 \par $I_{j}\left(t\right)$ and $H_{1}\left(t\right)$ are IG processes and hence their \chf's at time $t$ can be computed starting from the \chf\ expression of an IG \rv, which is reported in Appendix \ref{sec:Appendix}, whereas $Z_{a}\left(t\right)$ is the $a$-reminder of $H_{1}\left(t\right)$ and then its \chf\ can be easily computed relying on Equation \eqref{eqn:chfsdlaws}. The obtained \chf's of $H_{j}\left(t\right)$, $I_{j}\left(t\right)$ and $Z_{a}\left(t\right)$ are those of Equations \eqref{eqn:chfNigSem}.

Let be $\phi_{\boldsymbol{Y}\left(t\right)}\left(\boldsymbol{u}\right)\coloneqq \mathbb{E}\left[e^{i u_{1} Y_{1}\left(t\right)+ i u_{2} Y_{2}\left(t\right) }\right]$ the \chf\ of the process $\boldsymbol{Y}\left(t\right)$ defined in \eqref{eqn:subordinatedprocess}: conditioning on $G_{1}\left(t\right)$ and $G_{2}\left(t\right)$ and recalling that $W_{1}\left(t\right)$ and $W_{2}\left(t\right)$ are independent \BM's we get:
\begin{equation*}
    \phi_{\boldsymbol{Y}\left(t\right)}\left(\boldsymbol{u}\right) = \mathbb{E}\left[e^{i\left(u_{1}\mu_{1} + i \frac{\sigma_{1}^{2} u_{1}^{2}}{2}\right) G_{1}\left(t\right) } e^{i\left(u_{2}\mu_{2} + i \frac{\sigma_{2}^{2} u_{2}^{2}}{2}\right) G_{2}\left(t\right) }\right]
\end{equation*}
Substitute in the previous equation the expression of $G_{j}\left(t\right)$, given by \eqref{eqn:SSubordinators},  for $j=1,2$: by the property of the expected value for the product of independent \rv's, since $I_{j}\left(t\right),H_{1}\left(t\right)$ and $Z_{a}\left(t\right)$ are mutually independent processes, we finally get the result of the Equation \eqref{eqn:generalchf}.
\end{proof}

\subsection{Semeraro-Luciano's \sd-NIG model}
\label{sec:LSSD}
In this subsection we extend the model of \citet{SL2010} and we build bivariate NIG processes with stochastic delays relying on the \sd\ subordinators $\left(H_{1}\left(t\right),H_{2}\left(t\right)\right)$ defined in \eqref{eqn:subordinatorsSD}. Unlike the previous model, standard correlated \BM's, $W_{j}^{\rho}\left(t\right)$, are considered in order to obtain higher correlations in log-returns.

Let $I_{j}\left(t\right), \; j=1,2$, be subordinators and let $H_{1}\left(t\right)$ and $H_{2}\left(t\right)$ be two \sd\ subordinators independent of  $I_{j}\left(t\right)$. We define:

\begin{equation}
\boldsymbol{Y}^{\rho}\left(t\right) =
				\left(
				\begin{array}{ll}
				\mu_{1} I_{1}\left(t\right) + \sigma_{1}W_{1}\left(I_{1}\left(t\right)\right) +\alpha_{1}\mu_{1}H_{1}\left(t\right) + \sqrt{\alpha_{1}}\sigma_{1}W_{1}^{\rho}\left(H_{1}\left(t\right)\right) \\
\mu_{2} I_{2}\left(t\right) + \sigma_{2}W_{2}\left(I_{2}\left(t\right)\right) +\alpha_{2}\mu_{2}H_{2}\left(t\right) + \sqrt{\alpha_{2}}\sigma_{2}\left(W_{2}^{\rho}\left(aH_{1}\left(t\right)\right) + \tilde{W}\left(Z_{a}\left(t\right)\right)\right)
				
 				\end{array}\right)
				\label{eqn:LSgeneralization}
\end{equation}
where $W_{1}\left(t\right)$ and $W_{2}\left(t\right)$ are standard independent \BM's,
$\mathbb{E}\left[dW_{1}^{\rho}\left(t\right)dW_{2}^{\rho}\left(t\right)\right] = \rho dt$
and $\tilde{W}\left(t\right)$ is another standard \BM\ independent of $\boldsymbol{W}\left(t\right) = \left(W_{1}\left(t\right),W_{2}\left(t\right)\right)$ and $\boldsymbol{W}^{\rho}\left(t\right) = \left(W_{1}^{\rho}\left(t\right),W_{2}^{\rho}\left(t\right)\right)$.\\

\par A bivariate version of NIG process with \sd-subordinators can be easily obtained letting $H_{j}\left(t\right)$ and $I_{j}\left(t\right)$ for $j=1,2$ be distributed as in the previous section.
Moreover, the expression of the \chf\ of the process $\boldsymbol{Y}^{\rho}\left(t\right)$ at time $t$ is  given by the following proposition.
\begin{prop}
\label{prop:jointChfSemeraroLuciano}
The joint \chf\ $\phi_{\boldsymbol{Y}^{\rho}\left(t\right)}\left(\boldsymbol{u}\right)$ of the process $\boldsymbol{Y}^{\rho}\left(t\right) = \left(Y_{1}^{\rho}\left(t\right),Y_{2}^{\rho}\left(t\right)\right)$ at time $t$ defined in \eqref{eqn:LSgeneralization} is given by:

\begin{equation}
\begin{split}
\phi_{\boldsymbol{Y}\left(t\right)^{\rho}}\left(\boldsymbol{u}\right) = & \phi_{I_{1}\left(t\right)}\left(u_{1}\mu_{1} + \frac{i}{2}\sigma_{1}^{2}u_{1}^{2}\right) \phi_{I_{2}\left(t\right)}\left(u_{2}\mu_{2} + \frac{i}{2}\sigma_{2}^{2}u_{2}^{2}\right)\\
 & \phi_{H_{1}\left(t\right)}\left(\frac{i}{2}u_{1}^{2}\alpha_{1}\sigma_{1}^{2}\left(1-a\right) +  \boldsymbol{u}^{T}\boldsymbol{\mu} +\frac{i}{2}\boldsymbol{u}^{T}a\Sigma \boldsymbol{u}\right) \phi_{Z_{a}\left(t\right)}\left(u_{2}\mu_{2}\alpha_{2} + \frac{i}{2}u_{2}^{2}\alpha_{2}\sigma_{2}^{2}\right)
\end{split}
\nonumber
\end{equation}
where $\boldsymbol{\mu} = \left[\alpha_{1}\mu_{1},a\alpha_{2}\mu_{2}\right]$ and

\[\Sigma =
\begin{bmatrix}
	\alpha_{1}\sigma_{1}^{2} & \sqrt{\alpha_{1}\alpha_{2}}\sigma_{1}\sigma_{2}\rho \\
	\sqrt{\alpha_{1}\alpha_{2}}\sigma_{1}\sigma_{2}\rho & \alpha_{2}\sigma_{2}^{2}
\end{bmatrix} 
\]
where $\phi_{H_{1}\left(t\right)},\phi_{H_{2}\left(t\right)}$ and $\phi_{Z_{a}\left(t\right)}$ were defined in Proposition \ref{cor:2DNIGchfSemeraro}.
\end{prop}
\begin{proof}
The proof retraces the idea we used in the proof of Proposition \ref{cor:2DNIGchfSemeraro}, recalling, in addition, that the \chf\ $\varphi\left(\boldsymbol{t}\right)$ of a multivariated normal \rv\ with mean vector $\boldsymbol{\mu}$ and covariance matrix $\boldsymbol{\Sigma}$ is given by:
\begin{equation*}
    \varphi\left(\boldsymbol{t}\right) = \exp\left(i\boldsymbol{\mu}^{T}\boldsymbol{t} - \frac{1}{2}\boldsymbol{t}^{T}\boldsymbol{\Sigma}\boldsymbol{t}\right).
\end{equation*}

\end{proof}

It is easy to show, by direct computation or by using the \chf\ of Proposition \ref{prop:jointChfSemeraroLuciano}, that the linear correlation coefficient at time $t$ is given by:

\begin{equation}
	\rho_{\boldsymbol{Y}^{\rho}\left(t\right)} = \frac{a\left(\mu_{1}\mu_{2}\alpha_{1}\alpha_{2} A +\rho A \sigma_{1}\sigma_{2}\sqrt{\alpha_{1}\alpha_{2}}\right)}{\sqrt{\sigma_{1}^{2} + \mu_{1}^{2}\alpha_{1}}\sqrt{\sigma_{2}^{2} + \mu_{2}^{2}\alpha_{2}}}
	\label{eqn:linCorrSL}
\end{equation}

\noindent Once again, $a$ can be seen as the parameter that activates stochastic delay.

\subsection{Ballotta-Bonfiglioli's \sd-NIG model}
\label{sec:BBSD}
The construction of bivariate \Levy\ processes proposed by \citet{BB2013} is slightly different from those of \citet{Semeraro2008}  and \citet{SL2010} because the dependence is not introduced at the level of the subordinators but rather directly on the subordinated processes. Nevertheless, we can also extend this approach to include stochastic delay. \\
\par The construction of the a bivariated process with stochastic delay proceeds as follow. Let $H_{1}\left(t\right)$ and $H_{2}\left(t\right)$ be \sd\ subordinators of \eqref{eqn:subordinatorsSD}: define subordinated \BM's $R_{j}\left(t\right)$, for $j=1,2$, with drift $\beta_{R_{j}}\in \mathbb{R}$ and diffusion $\gamma_{R_{j}}\in \mathbb{R}^{+}$, as:
\begin{align}
R_{1}\left(t\right) & = \beta_{R_{1}}H_{1}\left(t\right) + \gamma_{R_{1}}W\left(H_{1}\left(t\right)\right) \nonumber \\
R_{2}\left(t\right) & = \beta_{R_{2}}H_{2}\left(t\right) + \gamma_{R_{2}}\left(W\left(aH_{1}\left(t\right)\right) + \tilde{W}\left(Z_{a}\left(t\right)\right)\right) \label{eqn:Rprocess}
\end{align}
where $W\left(t\right)$ and $\tilde{W}\left(t\right)$ are standard independent \BM's.
Let subordinated \BM's $X_{j}\left(t\right)$,  with drift $\beta_{j} \in \mathbb{R}$ and diffusion $\gamma_{j}\in \mathbb{R}^{+}$, be given by:
\begin{equation*}
X_{j}\left(t\right) = \beta_{j}G_{j}\left(t\right) + \gamma_{j}W_{j}\left(G_{j}\left(t\right)\right)
\end{equation*}
where $W_{j}\left(t\right)$ are standard independent \BM's whereas $G_{j}\left(t\right)$ are arbitrary subordinators with variance at time $t=1$ given by $\nu_{j} \in \mathbb{R}^{+}$.
\par Finally, combining previous processes, we can define the process $\boldsymbol{Y}\left(t\right)$ as follow:

\begin{equation}
\boldsymbol{Y}\left(t\right) = \left(Y_{1}\left(t\right),Y_{2}\left(t\right)\right) = \left(X_{1}\left(t\right) + a_{1}R_{1}\left(t\right), X_{2}\left(t\right) + a_{2} R_{2}\left(t\right)\right)
\label{eqn:extendedModelBB}
\end{equation}
where $a_{j} \in \mathbb{R}$.\\
\par As detailed in \citet{BB2013} and \citet{gardini2020}, for any chosen distribution for the margin process $Y_{j}\left(t\right)$, for example a NIG distribution, it is possible to impose convolution conditions on processes $X_{j}\left(t\right)$ and $R_{j}\left(t\right)$ so that their linear combination has the same given distribution of $Y_{j}\left(t\right)$. The following proposition shows how to build a bivariate NIG process with stochastic delays and gives the closed form expression for its \chf.  

\begin{prop}
Consider an IG subordinator $H_{1}\left(t\right) \sim IG_{T}\left(t,\frac{t^{2}}{\nu_{R}}\right)$,   $H_{2}\left(t\right)$ defined in Equation \eqref{eqn:subordinatorsSD} and 
 $R_{j}\left(t\right)$ given by \eqref{eqn:Rprocess}. Let then $X_{j}\left(t\right)$ be a subordinated \BM\ via an IG process $G_{j}\left(t\right) \sim IG_{T}\left(t,\frac{t^{2}}{\nu_{j}}\right)$, for $j=1,2$. 
 \\ Then the components $Y_{j}\left(t\right)$ in \eqref{eqn:extendedModelBB} are distributed according to a $NIG$ 
law and the joint \chf\ is
\begin{equation}
\begin{split}
\phi_{\boldsymbol{Y}\left(t\right)}\left(u_{1},u_{2}\right) =  \phi\left(\beta_{1}u_{1} + \frac{i}{2}u_{1}^{2}\gamma_{1}^{2}; t,\frac{t^{2}}{\nu_{1}}\right)
 \phi\left(\beta_{2}u_{2} + \frac{i}{2}u_{2}^{2}\gamma_{2}^{2}; t,\frac{t^2}{\nu_{2}}\right)
 \xi\left(\boldsymbol{a}\circ\boldsymbol{u}\right) 
\label{eqn:BBExtended}
\end{split}
\end{equation}
where $\phi\left(u;\mu,\lambda\right)$  is the \chf\ of a $IG_{T}\left(\mu,\lambda\right)$ distributed \rv, $\boldsymbol{a} = \left(a_{1},a_{2}\right)$, $\boldsymbol{u} = \left(u_{1},u_{2}\right)$ and $\circ$ is the Hadamard product. Finally $\xi\left(\boldsymbol{u}\right)$ is given by:

\begin{equation}
\begin{split}
\xi\left(\boldsymbol{w}\right) =& \phi\left(
w_{1}\beta_{R_{1}} + w_{2}\beta_{R_{2}}a + \frac{i}{2}\left(w_{1}^{2}\gamma_{R_{1}}^{2} +2w_{1}w_{2}\gamma_{R_{1}}\gamma_{R_{2}}a +  w_{2}^{2}a\gamma_{R_{2}}^{2}\right); t,\frac{t^2}{\nu_{R}}\right) \\
&  \frac{\phi\left(w_{2}\beta_{R_{2}} +\frac{i}{2}w_{2}^{2}\gamma_{R_{2}}^{2};t,\frac{t^2}{\nu_{R}}\right)}{\phi\left(w_{2}a\beta_{R_{2}} +\frac{i}{2}a^{2}w_{2}^{2}\gamma_{R_{2}}^{2};t,\frac{t^2}{\nu_{R}}\right)}
\end{split}
\label{eqn:Rprocesschf}
\end{equation}
\end{prop}

\begin{proof}
Relying on properties of the IG distribution in Appendix \ref{sec:Appendix}, it is easy to check that marginal distributions of $\boldsymbol{Y}\left(t\right)$ process have a NIG law.\\
\par Since $X_{1}\left(t\right)$, $X_{2}\left(t\right)$ and $\boldsymbol{R}\left(t\right)$ are mutually independent we have that
\begin{equation}
    \phi_{\boldsymbol{Y}\left(t\right)}\left(u_{1},u_{2}\right) = \mathbb{E}\left[e^{iu_{1}X_{1}\left(t\right)}\right] \mathbb{E}\left[e^{iu_{2}X_{2}\left(t\right)}\right]  \mathbb{E}\left[e^{iu_{1}R_{1}\left(t\right) + iu_{2}R_{2}\left(t\right)}\right]
    \label{eqn:BBproof1}
\end{equation}
\par The computation consists is two steps: firstly we compute the \chf\ $\mathbb{E}\left[e^{iu_{1}R_{1}\left(t\right) + iu_{2}R_{2}\left(t\right)}\right]$  
 of the joint process $\boldsymbol{R}\left(t\right)$ at time $t$ defined in \eqref{eqn:Rprocess}. This can be done by conditioning with respect $H_{1}\left(t\right)$ and $Z_{a}\left(t\right)$, relying upon the independence of $W\left(t\right)$ and $\tilde{W}\left(t\right)$ and recalling the expression of the \chf\ of a $IG_{T}\left(t,\frac{t^{2}}{\nu_{R}}\right)$ \rv, which is given in Appendix \ref{sec:Appendix}, and that of its $a$-reminder, obtained by applying the Equation \eqref{eqn:chfsdlaws}. By direct computation we obtain that the \chf\ of $\boldsymbol{R}\left(t\right)$ has the form shown in Equation \eqref{eqn:Rprocesschf} valuated at $\boldsymbol{w}=\boldsymbol{a}\circ \boldsymbol{u}$.
\par Secondly, we observe that first two terms of the right hand side of the Equation \eqref{eqn:BBproof1} are the \chf's of subordinated \BM's where subordinators are IG processes and hence their expressions are given by:
\begin{equation}
    \mathbb{E}\left[e^{iu_{j}X_{j}\left(t\right)}\right] = \phi\left(\beta_{j}u_{j} + \frac{i}{2}u_{j}^{2}\gamma_{j}^{2}; t,\frac{t^{2}}{\nu_{j}}\right)
        \label{eqn:BBproof2}
\end{equation}
where $\phi\left(u;\mu,\lambda\right)$ denotes the \chf\ of a \rv\ with $IG_{T}\left(\mu,\lambda\right)$ law.
Combining Equations \eqref{eqn:Rprocesschf}, \eqref{eqn:BBproof1} and     \eqref{eqn:BBproof2} we finally obtain \eqref{eqn:BBExtended}.
\end{proof}

\vspace{0.1cm}

\par The linear correlation coefficient of a bivariate \sd-NIG process at time $t$ can be directly computed and it is given by:

\begin{equation}
\rho_{\boldsymbol{Y}\left(t\right)} = \frac{a_{1}a_{2}a\left(\beta_{R_{1}}\beta_{R_{2}}\nu_{R}+ \gamma_{R_{1}}\gamma_{R_{2}}\right) }{\sqrt{\sigma_{1}^{2} + \mu_{1}^{2}\alpha_{1}}\sqrt{\sigma_{2}^{2} + \mu_{2}^{2}\alpha_{2}}}
\label{eqn:linCorrBB}
\end{equation}

As expected, if $a=1$ we retrieve the original expression of correlation coefficient obtained by \citet{BB2013}.





%% file: SimulationAlgorithm.tex
Simulating the paths of the model dynamics defined in Section \ref{sec:SDMODELS} can be accomplished by simulating BM's on a stochastic time grid generated by the relative IG \sd\ subordinators. These subordinators are only marginally IG, in order to get the joint trajectories one has to simulate the skeleton of $Z_a(t)$ in \eqref{eqn:subordinatorsSD} and therefore must have a way to draw from the law of the \arem\  $Z_{a}$ of an IG distribution.      

The methodology that we propose in this section is based on the close relation between \sd\ laws and \Levy-driven OU processes.
Following the naming convention in \citet{BNSh01} we say that a \Levy-driven OU process $X\left(t\right)$ is a IG-OU process if its stationary law is an $IG_{B}$ distribution with scale parameter $\delta$ and shape parameter $\gamma$. Now
a well known result (see for instance \citet{CT2003} or \citet{Sato}) is that, a
given one-dimensional distribution $D$ always is the stationary law of a suitable \Levy-driven OU process if and only if $D$ is \sd. As shown by \citet{Halgreen1979} the IG law  is \sd\ and can be taken as the stationary distribution of a fully-fledged OU process.

We recall that a \Levy-driven OU process is defined as,
\begin{equation}\label{eq:levy:ou}
X\left(t\right) = X\left(0\right)e^{-\lambda t} + \int_{0}^{t}e^{-\lambda \left(t-u\right)}dL\left(u\right)
\end{equation}
where $L\left(t\right)$ is a \Levy\ process and $\lambda > 0$. In addition, as observed in \citet{BNSh01}, $X(t)$ is stationary if and only if the \chf\ $\phi_X\left(u\right)$ of its marginal distribution is of the form
\begin{equation*}
    \phi_X(u) = \phi_X(ue^{-\lambda t})\chi_a(u, t)
\end{equation*}
where $\chi_a(u, t)$ is the \chf\ of the second term of \eqref{eq:levy:ou}. On the other hand, due to the definition of \sd, the last equation means that $\chi_a(u, t)$ is the \chf\ of the \arem\ of the stationary law if one sets $a=e^{-\lambda t}$. We can then write
\begin{equation}
X(t) = X\left(0\right)e^{-\lambda t} + Z_{e^{-\lambda t}}(t).
\label{eqn:OUsol1}
\end{equation}
Note that the parameter $e^{-\lambda t}$ is now time-dependent and the law of $Z_{e^{-\lambda t}}(t)$ coincides with that of $Z_a(t)$ with $a=e^{-\lambda t}$ only at a given time $t$, indeed $Z_{e^{-\lambda t}}(t)$ is not \Levy\ but rather an additive process. Nevertheless, in practice the simulation of the skeleton of a IG-OU process relies on the generation of a \rv\ that is distributed according to the law of the \arem\ of the stationary distribution setting $a=e^{-\lambda t}$. 

Starting from the results of \citet{Zhang2008} relative to IG-OU processes, we derive an efficient algorithm to simulate the \arem\ of the $IG_{B}(\delta, \gamma)$, that is the building block for the generation of the trajectory of the process $Z_a(t)$.
\begin{thm}[\citet{Zhang2008}]
\label{thm:Zhang}
The \rv\  
\begin{equation*}
Z_{a}^{\Delta} = \int_{0}^{\Delta}e^{-\lambda \left(\Delta -u\right)}dL\left(u\right), \quad a=e^{-\lambda \Delta}, \quad \Delta > 0  
\end{equation*}
  can be represented as  
\begin{equation*}
    Z_{a}^{\Delta} \eqd W_{0}^{\Delta} + \sum_{i=1}^{\tilde{N}^{\Delta}}W_{i}^{\Delta}
\end{equation*}
where $W_{0}^{\Delta} \sim IG_{B}\left(\delta\left(1-e^{-\frac{1}{2}\lambda \Delta}\right),\gamma\right)$,   $\tilde{N}^{\Delta}$ is a Poisson-distributed \rv\ with parameter $\delta\left(1-e^{-\frac{1}{2}\lambda \Delta}\right)\gamma$ and $W_{i}^{\Delta}$ are independent \rv's with \pdf:

\begin{equation}
    f_{W^{\Delta}}\left(w\right) = \frac{\gamma^{-1}}{\sqrt{2\pi}}w^{-\frac{3}{2}}\left(e^{\frac{1}{2}\lambda \Delta}-1\right)^{-1} \left(e^{-\frac{1}{2}\gamma^{2}w}-e^{-\frac{1}{2}\gamma^{2}we^{\lambda \Delta}}\right) \mathbbm{1}_{\left\{w>0\right\}}\left(w\right)
    \label{eqn_widensity}
\end{equation}
\end{thm}
Assuming for simplicity $\Delta=1$, we can then rely on Theorem \ref{thm:Zhang} to conceive the simulation procedure of two correlated IG \rv's with linear correlation coefficient $a$ and hence of the \sd\ subordinators of \eqref{eqn:subordinatorsSD} simply setting $\lambda = -\log a$. We get:

\begin{equation*}
Z_{a} \eqd W_{0} + \sum_{i=1}^{\tilde{N}}W_{i}
\end{equation*}
where $W_{0}\sim IG\left(\delta\left(1 - a^{\frac{1}{2}}\right),\gamma\right)$ and $\tilde{N} \sim Poisson\left(\delta\left(1 - a^{\frac{1}{2}}\right)\gamma\right)$.

 Drawing from  IG and Poisson laws is relatively easy, whereas the simulation of $W_{i}$ is non-standard and can be generated using the acceptance-rejection algorithm proposed by \citet{Zhang2008} observing that:

\begin{equation*}
    f_{W}\left(w\right) \le c\cdot \Gamma\left(\frac{1}{2},\frac{1}{2}\gamma^{2}\right)
\end{equation*}
where $c = \frac{1}{2}\left(1 + e^{\frac{1}{2}\lambda}\right)$ and $\Gamma(\alpha, \beta)$ denote the law of a gamma \rv\ with shape $\alpha>0$ and rate $\beta>0$.\\

\par Although \citet{Zhang2008} has illustrated a more accurate solution to reduce the expected number of iterations before acceptance $c$, acceptance-rejection algorithms might be slow and then sometimes inadequate for real time applications. This situation is exacerbated if the software implementation relies on interpreted languages like MATLAB, Python or R.  
In the following, we  detail a simple and more efficient way to draw from the \pdf\ $f_{W^{\Delta}}\left(w\right)$ without relying on acceptance-rejection methods. 

\par Assuming once again $\Delta  = 1$ and $\lambda = -\log a$, equation \eqref{eqn_widensity} becomes: 
\begin{equation*}
    f_{W}\left(w\right) = \frac{\gamma^{-1}}{\sqrt{2\pi}}w^{-\frac{3}{2}}\left(a^{-\frac{1}{2}}-1\right)^{-1} \left(e^{-\frac{1}{2}\gamma^{2}w}-e^{-\frac{1}{2}\gamma^{2}\frac{w}{a}}\right) \mathbbm{1}_{\left\{w>0\right\}}\left(w\right).
\end{equation*}
We recall that a \rv\ is distributed according to a Gamma law with shape $\alpha>0$ and rate $\beta>0$ if its \pdf\ is:
\begin{equation*}
    f\left(x\right) = \frac{\beta^{\alpha}}{\Gamma\left(\alpha\right)}x^{\alpha - 1}e^{-\beta x}
\end{equation*}
where $\Gamma\left(z\right) = \int_{0}^{\infty}x^{z-1}e^{-x}dx$ is the Euler Gamma function.
Knowing that $\Gamma\left(\frac{1}{2}\right) = \sqrt{\pi}$ and observing that:
\begin{equation*}
    \int_{1}^{\frac{1}{a}} e^{-\frac{\gamma^{2}}{2}wy} \frac{\gamma^{2}}{2} w dy = e^{-\frac{\gamma^{2}}{2}w} - e^{-\frac{\gamma^{2}}{2}\frac{w}{a}}
\end{equation*}
we can write:

\begin{equation*}
\begin{split}
   f_{W}\left(w\right) & = \int_{1}^{\frac{1}{a}} \frac{y^{-\frac{1}{2}}}{2\left(a^{-\frac{1}{2}}-1\right)} \cdot \frac{\left(\frac{\gamma^{2}}{2} y\right)^{\frac{1}{2}}w^{-\frac{1}{2}}e^{-\frac{\gamma^{2}}{2} y w}}{\Gamma\left(\frac{1}{2}\right)}dy \\ & = \int_{1}^{\frac{1}{a}} f_{Y}\left(y\right) \cdot f_{\Gamma}\left(w\Big|\alpha = \frac{1}{2},\beta = \frac{\gamma^{2}}{2} y\right) dy
\end{split}
\end{equation*}
This means that $f_{W}\left(w\right)$ is a mixture of a Gamma law  $\Gamma\left(\alpha = \frac{1}{2},\beta = \frac{\gamma^{2}}{2} y\right)$  and a law whose \pdf\ and \cdf\ are respectively:

\begin{align*}
f_{Y}\left(y\right) &= \frac{y^{-\frac{1}{2}}}{2\left(a^{-\frac{1}{2}}-1\right)} \mathbbm{1}_{1  \le y \le \frac{1}{a}} \\
F_{Y}\left(y\right) &= \frac{y^{\frac{1}{2}} -1}{a^{-\frac{1}{2}}-1}\mathbbm{1}_{1  \le y \le \frac{1}{a}}
\end{align*}
The simulation of $Z_{a}$ and of the \rv\ $Y$ distributed according to the law with \cdf\ $F_Y(y)$ is  straightforward as is summarized in Algorithms \ref{alg:ZaSimulation} and \ref{alg:WiPiergiacomo}, respectively.

\begin{algorithm}
\caption{Simulation of $Z_{a}$}
\label{alg:ZaSimulation}
\begin{algorithmic}[1]
\State Simulate $W_{0} \sim IG\left(\delta\left(1-\sqrt{a}\right),\gamma\right)$
\State Simulate $\tilde{N}\sim Poisson\left(\delta\left(1-\sqrt{a}\right)\gamma\right)$
\State Simulate $W_{i}, i=1\dots \tilde{N}$ using Algorithm \ref{alg:WiPiergiacomo}
\State Set $Z_{a} = \sum_{i=0}^{\tilde{N}}W_{i}$
\end{algorithmic}
\end{algorithm}

\begin{algorithm}
\caption{Simulation of $W_{i}, \tilde{N}$}
\label{alg:WiPiergiacomo}
\begin{algorithmic}[1]
\State Simulate $U_{i} \sim U\left(\left[0,1\right]\right)$
\State Compute $Y_{i} = \left(1 + \left(a^{-\frac{1}{2}}-1\right)U_{i}\right)^{2}$
\State Simulate $W_{i}$ from a $\Gamma\left(\frac{1}{2},\frac{1}{2}\gamma^{2}Y_{i}\right)$
\end{algorithmic}
\end{algorithm}

In Table \ref{tab:momentscomparison} we compare theoretical values of the first five moments of $Z_{a}$ against those obtained by MC simulations using Algorithm \ref{alg:ZaSimulation}. We observe that the precision of the algorithms is good for different values of $a\in \left(0,1\right)$. In Figure \ref{fig:distroandScatter} we draw the probability density function of two correlated \rv\ $X,Y \sim IG_{B}\left(\delta,\gamma\right)$ and their scatter plot for two different values of $a$.

\begin{table}

\parbox{.5\linewidth}{
\centering
\begin{tabular}{cccc}
\toprule
$\mathbb{E}\left[Z_{a}^{n}\right]$ & $T$  & $N$\\
\midrule
$\mathbb{E}\left[Z_{a}^{1}\right]$ & 3.00 & 3.00\\
$\mathbb{E}\left[Z_{a}^{2}\right]$ & 10.47 & 10.48\\
$\mathbb{E}\left[Z_{a}^{3}\right]$ & 42.17 & 42.26\\
$\mathbb{E}\left[Z_{a}^{4}\right]$ & 194.72 & 195.49\\
$\mathbb{E}\left[Z_{a}^{5}\right]$ & 1021.84 & 1029.41\\
\bottomrule
\end{tabular}
\subcaption{$a=0.1$}
}
\hfill
\parbox{.5\linewidth}{
\centering
\begin{tabular}{cccc}
\toprule
$\mathbb{E}\left[Z_{a}^{n}\right]$ & $T$ & $N$\\
\midrule
$\mathbb{E}\left[Z_{a}^{1}\right]$ & 1.67 & 1.67\\
$\mathbb{E}\left[Z_{a}^{2}\right]$ & 3.89 & 3.89\\
$\mathbb{E}\left[Z_{a}^{3}\right]$ & 11.91 & 11.90\\
$\mathbb{E}\left[Z_{a}^{4}\right]$ & 45.58 & 45.46\\
$\mathbb{E}\left[Z_{a}^{5}\right]$ & 209.90 & 208.97\\
\bottomrule
\end{tabular}
\subcaption{$a=0.5$}
}
\hfill
\parbox{.5\linewidth}{
\centering
\begin{tabular}{cccc}
\toprule
$\mathbb{E}\left[Z_{a}^{n}\right]$ & $T$ & $N$\\
\midrule
$\mathbb{E}\left[Z_{a}^{1}\right]$ & 1.00 & 1.00\\
$\mathbb{E}\left[Z_{a}^{2}\right]$ & 1.76 & 1.76\\
$\mathbb{E}\left[Z_{a}^{3}\right]$ & 4.56 & 4.59\\
$\mathbb{E}\left[Z_{a}^{4}\right]$ & 15.77 & 15.89\\
$\mathbb{E}\left[Z_{a}^{5}\right]$ & 67.94 & 68.66\\
\bottomrule
\end{tabular}
\subcaption{$a=0.7$}
}
\hfill
\parbox{.5\linewidth}{
\centering
\begin{tabular}{cccc}
\toprule
$\mathbb{E}\left[Z_{a}^{n}\right]$ & $T$  & $N$\\
\midrule

$\mathbb{E}\left[Z_{a}^{1}\right]$ & 0.33 & 0.33\\
$\mathbb{E}\left[Z_{a}^{2}\right]$ & 0.39 & 0.40\\
$\mathbb{E}\left[Z_{a}^{3}\right]$ & 0.85 & 0.86\\
$\mathbb{E}\left[Z_{a}^{4}\right]$ & 2.66 & 2.68\\
$\mathbb{E}\left[Z_{a}^{5}\right]$ & 10.71 & 10.72\\
\bottomrule
\end{tabular}
\subcaption{$a=0.9$}
}
\caption{Moments comparison using $N_{sim} = 10^{6}$ for $\delta = 5$ and $\gamma = 1.5$. $T$ stands for the values of the theoretical n-th moment, whereas $N$ stands for the MC-based estimations.}
\label{tab:momentscomparison}
\end{table}

\begin{figure}
    \centering
    \includegraphics[scale = 0.3]{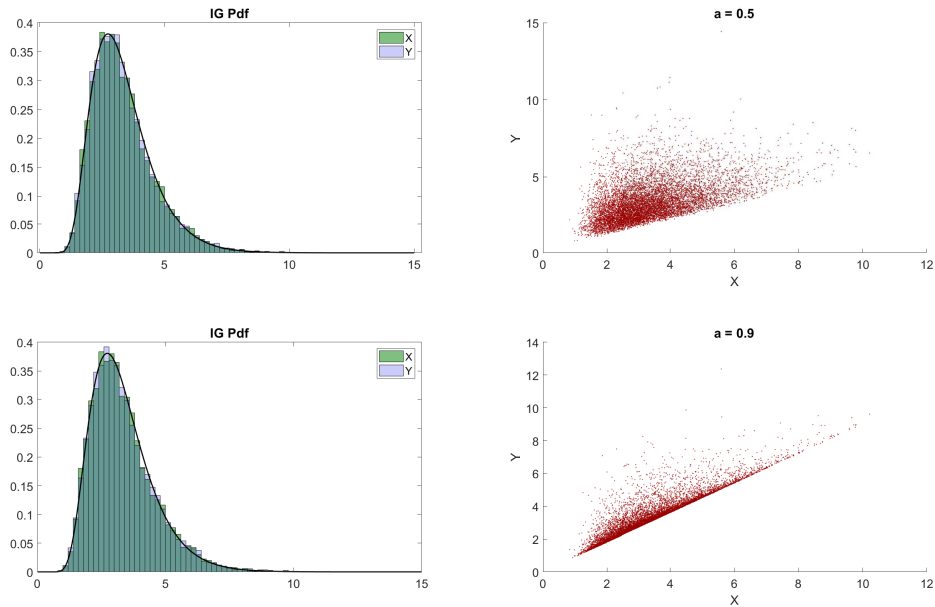}
    \caption{Correlated \rv\ $X$ and $Y$ for $\delta = 5$ and $\gamma = 1.5$ and their scatter plots for $a=0.5$ and $a=0.9$.}
    \label{fig:distroandScatter}
\end{figure}

The proposed algorithm is extremely fast as one can see from results reported in Table \ref{tbl:speed}. This time complexity analysis was implemented on a PC having an Intel Core i5-10210U 2.11 GHz processor.

\begin{table}
\centering
\begin{tabular}{ccccc}
\toprule
$N_{sim}$ & $10^{3}$  & $10^{4}$ & $10^{5}$ & $10^{6}$ \\
\midrule
Time (s) & $1.05 \cdot 10^{-5}$  & $6.54 \cdot 10^{-5}$ & $6.98 \cdot 10^{-4}$ & $7.48 \cdot10^{-3}$  \\
\bottomrule
\end{tabular}
\caption{Average computational time on one hundred runs of Algorithm \ref{alg:WiPiergiacomo} varying the number of simulations.}
\label{tbl:speed}
\end{table}

\vskip .3cm
The simulation of the $a$-reminder of an IG law provides the generation of the joint trajectories of the \sd\ subordinators $H_{1}\left(t\right),H_{2}\left(t\right)$ and therefore of the models presented in Section \ref{sec:SDMODELS}. The application of these MC schemes will be shown in next section.

%% file: NumericalResults.tex
In this section we use the bivariate \Levy\ processes illustrated in Section \ref{sec:SDMODELS} to model power and gas forward markets.

Following \citet{CT2003}, we assume that each forward price dynamics is driven by an exponential \Levy\ process based  on $Y_j\left(t\right), j=1, 2$ derived in Section \ref{sec:SDMODELS}. The forward price $F_{j}(t),\; j=1,2$ at time $t$ can be defined as follow:
\begin{equation}
F_{j}\left(t\right) = F_{j}\left(0\right) e^{\omega_{j} t + Y_{j}\left(t\right) }
\label{eqn:themodel}
\end{equation}
where $\omega_{j}$ is the drift correction required for risk-neutral arguments such that
\begin{equation}
\omega_{j} = -\varphi_{j}\left(-i\right)
\label{eqn:omegaexpr}
\end{equation}
where $\varphi_{j}\left(u\right)$ is the characteristic exponent of the process $Y_{j}\left(t\right)$.\\
\par In order to calibrate our model we use the two steps  procedure adopted in \citet{SL2010} and in \citet{gardini2020}: since the marginal distributions do not depend on the parameters used to model the structure of dependence one can firstly fit the marginal parameters on quoted vanilla product and, secondly, dependence ones on forward historical data. The choice of fitting the dependence structure on historical quotations is motivated by the fact that derivative contracts written on more than one underlying asset are extremely illiquid.\\

\noindent  Once calibrated the marginal parameters, we consider spread options written on future prices, which payoff is given by 

\begin{equation*}
\Phi_{T} = \left(F_{1}\left(T\right) - F_{2}\left(T\right) - K\right)^{+}
\end{equation*}
can be priced. It customary to reserve the name  \textit{Cross-Border} or \textit{Spark-Spread} option if the futures are relative to power or gas markets, respectively.
In all experiments we use the MC technique with $N_{sim}=10^{6}$ simulations and the Fourier-based method presented by \citet{caldanafusai2016}. This method provides a good approximation for spread-options prices and it's simpler to implement than the one proposed by \citet{hurd2009}, because it requires only one Fourier inversion. 

The numerical investigation is split into two parts:
in the first one we use \sd-NIG processes to model German and French power forward markets, whereas in the second part we focus on German power and natural gas forward markets.
\par All these markets are very correlated in particular, the German and French power forward markets exhibit an extremely high log-returns correlation. This is due to the structure of the electricity network that connects the two  countries and to the fact that electricity cannot be stored. Therefore, if the price of electricity rises in Germany we can observe an increase  of electricity prices in French as well. The log-returns correlation between German power  and natural gas forward markets is still positive but lower than that of the previous case. This depends on the percentage of installed capacity depending on natural gas (in 2020, 13.9\% in Germany) and, moreover, gas can be stored. For the sake of concision we introduce the following notation:

\begin{itemize}
\item (\emph{SSD - NIG}): \sd-NIG model presented in Section \ref{sec:NIGSemeraro}.
\item (\emph{LSSD - NIG}): \sd-NIG model presented in Section \ref{sec:LSSD}.
\item (\emph{BBSD - NIG}): \sd-NIG model presented in Section \ref{sec:BBSD}.
\end{itemize}


\input{NumericalResultsPower}

\subsection{Application to German Power market and NCG Gas Market}
\label{sec:NIGtoSparkSpread}
In this section we present numerical results obtained applying our models to German power forward market (DE) and to natural gas forward market (NCG). These two markets are positively correlated,  but the log-return correlation is lower that the one between power futures. \\
The data-set\footnote{Data Source: www.eex.com and www.theice.com} we relied upon is the following one:

\begin{itemize}
\item Forward quotations from 1 July 2019 to 09 September 2019 relative to the Month January 2020 for the Power Forward in Germany and  the Gas NCG Forward.
\item Call Options on power forward NCG with settlement date 9 September 2019. As done before, we use strike prices $K$ in a range of $\pm 10\, [EUR/MWh]$ around the settlement price of the forward contract.
\item We assume a risk-free rate $r=0.015$.
\item The historical correlation between log-returns is $\rho_{mkt} = 0.54$.
\end{itemize}
 In the picture at the bottom of Figure \ref{fig:SparkSpreadOption} we observe that all models provide a good fitting of quoted market options because the relative error $\epsilon_{i}$ is small. 
The picture at the top of Figure \ref{fig:SparkSpreadOption} shows that the \emph{SSD-NIG} model overprices the Spark-Spread option due to the fact that fitted model correlation is close to zero, as shown by the value $\rho_{mod}$ in Table \ref{table:depparamsSSDSS}. In contrast, \emph{LSSD-NIG} and \emph{BBSD-NIG} models provide a lower price and catch the right level of market correlation as shown in Tables \ref{table:depparamsLSSDSS},\ref{table:depparamsBBSDSS}. We can conclude that both \emph{LSSD-NIG} and \emph{BBSD-NIG} models can be used to price \textit{Spark-Spread} options.
Table \ref{tbl:commoparams} shows fitted common parameters whereas  dependence parameters for \emph{SSD-NIG}, \emph{LSSD-NIG} and \emph{BBSD-NIG} models are shown in Tables \ref{table:depparamsSSDSS}, \ref{table:depparamsLSSDSS}, \ref{table:depparamsBBSDSS}: the value of $a$, the \sd\ parameter which aims to model the \textit{stochastic delay}, is shown in Table \ref{tbl:sdparamsval}.
The value is still close to one but it is smaller than that estimated for the power forward markets. From the expressions of the linear correlation coefficient reported in equations \eqref{eqn:linCorrS}, \eqref{eqn:linCorrSL} and \eqref{eqn:linCorrBB}, it is easy to see that a change in the value of $a$ has an impact on the value of the correlation coefficient and it is a matter of fact that even a small change in correlation has a high impact on the spread option price. On the other hand, unlike electricity, natural gas can be stored and therefore the impact on on the power market  can be moderated and delayed, for example, using storage contracts or other types of OTC derivatives. If the gas price suddenly rises then it is not rare to observe that electricity price is not immediately effected.

\begin{table}[!htb]
\scriptsize

\begin{minipage}{1\linewidth}
\centering
    \begin{tabular}[t]{ccccccc}
    \toprule
    Model & $\mu_{1}$ & $\mu_{2}$ & $\sigma_{1}$ & $\sigma_{2}$ & $\alpha_{1}$ & $\alpha_{2}$  \\ [0.5ex]
    \midrule
\emph{SSD} & 0.37 & 0.20  & 0.44 & 0.33 & 0.09 & 0.07   \\
\emph{LSSD} & 0.37 & 0.20  & 0.44 & 0.33 & 0.09 & 0.07   \\
\emph{BBSD} & 0.37 & 0.20  & 0.44 & 0.33 & 0.09 & 0.07   \\
\bottomrule
\end{tabular}
\caption{Fitted marginal parameters for German and French power markets.}
\label{tbl:commoparams}
\end{minipage}\hfill
\vspace{0.2cm}
\begin{minipage}{.23\linewidth}
\centering

\begin{tabular}[t]{cc}
\toprule
Parameter & Value \\ [0.5ex]
\midrule
$A$ & 11.27  \\
$B$ & 1.00  \\
$a$ & 0.99  \\
$\rho_{mod}$ & 0.03  \\ [1ex]
\bottomrule\end{tabular}
\caption{\emph{SSD}}
\label{table:depparamsSSDSS}
\end{minipage}\hfill
\begin{minipage}{.3\linewidth}
\centering

\smallskip

\begin{tabular}[t]{cc}
\toprule
Parameter & Value \\ [0.5ex]
\midrule
$A$ & 8.79  \\
$B$ & 1.00  \\
$\rho$ & 0.87  \\
$a$ & 0.90  \\
$\rho_{mod}$ & 0.54  \\ [1ex]
\bottomrule
\end{tabular}
\caption{\emph{LSSD}}
\label{table:depparamsLSSDSS}
\end{minipage}\hfill
\begin{minipage}{.4\linewidth}
\centering

\smallskip

\begin{tabular}[t]{cccc}
\toprule
Parameter & Value & Parameter & Value\\ [0.5ex]
\midrule
$\beta_{1}$ & 0.11 & $\beta_{R_{2}}$ & 0.23 \\
$\beta_{2}$ & 0.09  & $\gamma_{R_{1}}$ & 0.56\\
$\gamma_{1}$ & 0.24 & $\gamma_{R_{2}}$ & 0.50  \\
$\gamma_{2}$ & 0.22 & $\nu_{R}$ & 0.13  \\
$\nu_{1}$ & 0.28 & $a$ & 0.89  \\
$\nu_{2}$ & 0.15 & $\rho_{mod}$ & 0.54  \\
$\beta_{R_{1}}$ & 0.38 & &\\ [1ex]
\bottomrule
\end{tabular}
\caption{\emph{BBSD}}
\label{table:depparamsBBSDSS}
\end{minipage}
\end{table}

\begin{table}[ht]
\scriptsize
\centering 
\begin{tabular}{cc}
\toprule
Model & $a$ \\ [0.5ex]
\midrule 
\emph{SSD} & 0.99  \\ 
\emph{LSSD} & 0.90  \\ 
\emph{BBSD} & 0.89  \\ [1ex] 
\bottomrule
\end{tabular}
\caption{Values for the $a$ parameter of three models.} 
\label{tbl:sdparamsval} 
\end{table}

\begin{figure}
    \centering
    \includegraphics[scale=0.25]{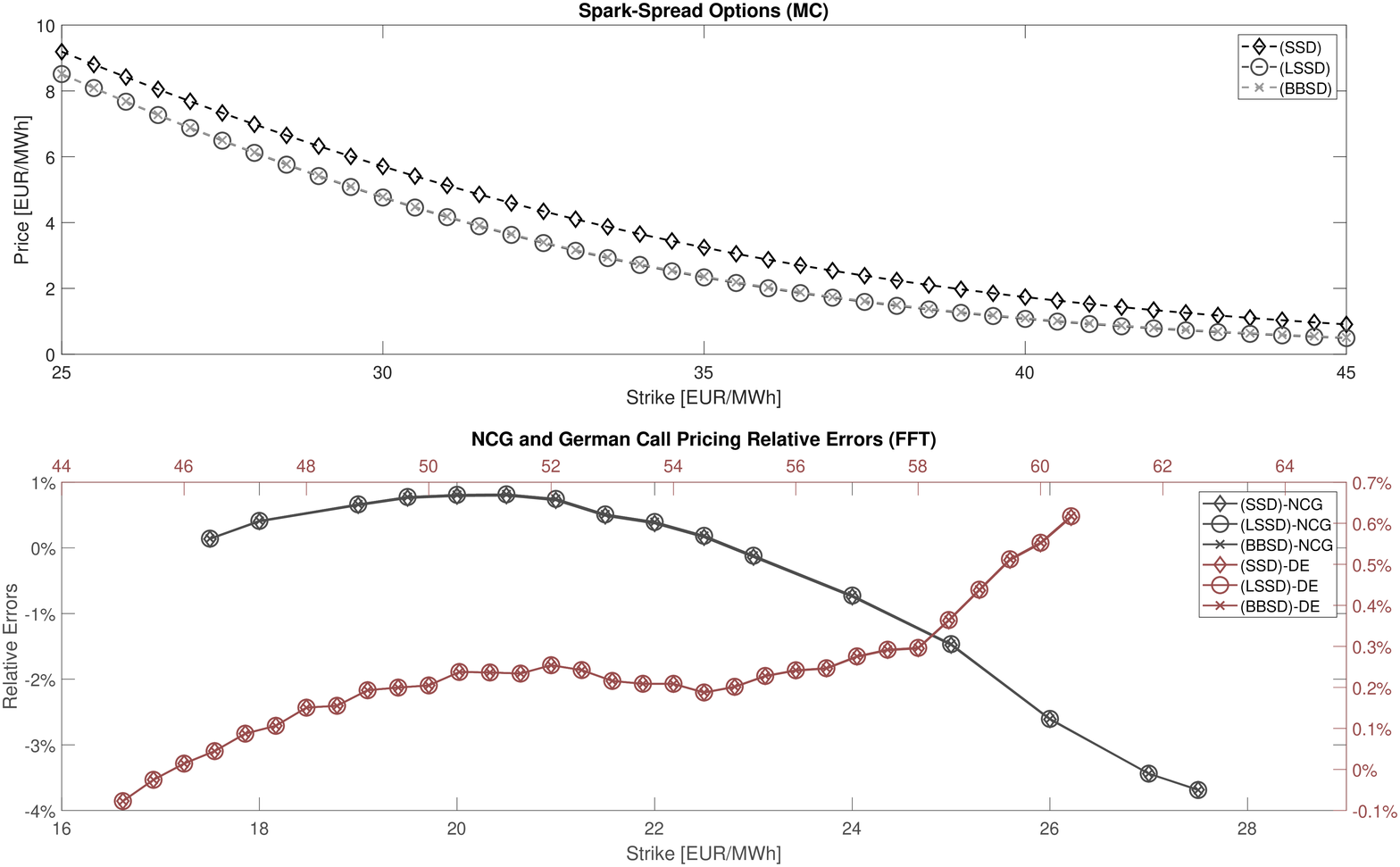}
    \caption{Percentage errors and Spark-Spread option prices.}
    \label{fig:SparkSpreadOption}
\end{figure}

%% file: NumericalResultsPower.tex

\subsection{Application to German and French Power Markets}
\label{sec:NIGtoCrossBorder}
In order to calibrate the proposed \sd-NIG models we rely upon derivative contracts written on the forward price of each underlying and upon the joint historical time series of forward quotations. The data-set\footnote{Data Source: www.eex.com.} is composed as follow:
\begin{itemize}
\item Forward quotations from 25 April 2017 to 12 November 2018 of Calendar 2019 power forward. Calendar power forward in German and France are stated respectively with DEBY and F7BY.
\item Call Options on power forward 2019 quotations for both countries with settlement date 12 November 2018. We used strikes in a range of $\pm 10\, [EUR/MWh]$ around the settlement price of the forward contract.
\item We assume a risk-free rate $r=0.015$.
\item The historical correlation observed between markets is $\rho_{mkt} = 0.94$.
\end{itemize}

We denote $\left(\theta_{1},\theta_{2}\right)$ parameters related to the French and German power forward markets respectively. Defining the error $\epsilon_{i}$ as
\begin{equation*}
\epsilon_{i} = \frac{C_{i}^{\theta}\left(K,T\right) - C_{i}}{C_{i}},
\end{equation*}
where $C_{i}^{\theta}\left(K,T\right)$ is the value of the $i$-th Call option obtained by the model and $C_{i}$ is its market price, the picture at the bottom of Figure \ref{fig:CrossBorderOption} shows that all models provide a good fit for quoted market options because $\epsilon$  is negligible. In Figure \ref{fig:CrossBorderOption} the picture at the top shows that the \emph{SSD-NIG} model overprices \textit{Cross-Border} options: this is because the fitted model correlation is low, as shown by the value $\rho_{mod}$ in Table \ref{table:depparamsSSDCB}, so one should avoid using this model for pricing. For \emph{LSSD-NIG} model the situation is better but it is not really able to capture the prevailing market correlation. Fortunately \emph{BBSD-NIG} model can replicate the market correlation and then can be used to price \textit{Cross-Border} options. Fitted common parameters are shown in Table \ref{tbl:commoparamsCB}, whereas the dependence parameters for \emph{SSD-NIG}, \emph{LSSD-NIG} and \emph{BBSD-NIG} models are shown in Tables \ref{table:depparamsSSDCB}, \ref{table:depparamsLSSDCB}, \ref{table:depparamsBBSDCB}, respectively. The value of $a$,  is shown in Table \ref{tbl:sdparamsvalCB}. We observe that the parameter $a$ is very close to one, as one should expected. Indeed this result has a very natural economic interpretation: the European electricity network is strongly connected and a price movement in either the German or French market one is propagated without  stochastic delay. Finally, in Table \ref{tbl:CrossBorderoptionpricesFFTvsMC} we compare values of \textit{Cross-Border} options priced using the FFT method proposed by \citet{caldanafusai2016} and the MC scheme we proposed in Section \ref{sec:SimAlg}. Option prices provided by both algorithms are very close and this allows us to use indistinctly FFT or MC method.

\begin{table}[!htb]
\scriptsize

\begin{minipage}{1\linewidth}
\centering
    \begin{tabular}[t]{ccccccc}
    \toprule
    Model & $\mu_{1}$ & $\mu_{2}$ & $\sigma_{1}$ & $\sigma_{2}$ & $\alpha_{1}$ & $\alpha_{2}$  \\ [0.5ex]
    \midrule
\emph{SSD} & 0.64 & 0.40  & 0.31 & 0.32 & 0.02 & 0.03   \\
\emph{LSSD} & 0.64 & 0.40  & 0.31 & 0.32 & 0.02 & 0.03   \\
\emph{BBSD} & 0.64 & 0.40  & 0.31 & 0.32 & 0.02 & 0.03   \\
\bottomrule
\end{tabular}
\caption{Fitted marginal parameters for German and French power markets.}
\label{tbl:commoparamsCB}
\end{minipage}\hfill
\vspace{0.2cm}
\begin{minipage}{.23\linewidth}
\centering

\begin{tabular}[t]{cc}
\toprule
Parameter & Value \\ [0.5ex]
\midrule
$A$ & 40.15  \\
$B$ & 1.00  \\
$a$ & 0.99  \\
$\rho_{mod}$ & 0.05  \\ [1ex]
\bottomrule\end{tabular}
\caption{\emph{SSD}}
\label{table:depparamsSSDCB}
\end{minipage}\hfill
\begin{minipage}{.3\linewidth}
\centering

\smallskip

\begin{tabular}[t]{cc}
\toprule
Parameter & Value \\ [0.5ex]
\midrule
$A$ & 40.15  \\
$B$ & 1.00  \\
$\rho$ & 0.99  \\
$a$ & 0.99  \\
$\rho_{mod}$ & 0.88  \\ [1ex]
\bottomrule
\end{tabular}
\caption{\emph{LSSD}}
\label{table:depparamsLSSDCB}
\end{minipage}\hfill
\begin{minipage}{.4\linewidth}
\centering

\smallskip

\begin{tabular}[t]{cccc}
\toprule
Parameter & Value & Parameter & Value\\ [0.5ex]
\midrule
$\beta_{1}$ & -0.001 & $\beta_{R_{2}}$ & 0.800 \\
$\beta_{2}$ & 0.013  & $\gamma_{R_{1}}$ & 0.448\\
$\gamma_{1}$ & 0.002 & $\gamma_{R_{2}}$ & 0.50  \\
$\gamma_{2}$ & 0.103 & $\nu_{R}$ & 0.025  \\
$\nu_{1}$ & 1.007 & $a$ & 0.99  \\
$\nu_{2}$ & 0.091 & $\rho_{mod}$ & 0.94  \\
$\beta_{R_{1}}$ & 0.554 & &\\ [1ex]
\bottomrule
\end{tabular}
\caption{\emph{BBSD}}
\label{table:depparamsBBSDCB}
\end{minipage}
\end{table}

\begin{table}[ht]
\scriptsize
\centering 
\begin{tabular}{cc}
\toprule
Model & $a$ \\ [0.5ex]
\midrule 
\emph{SSD} & 0.99  \\ 
\emph{LSSD} & 0.99  \\ 
\emph{BBSD} & 0.99  \\ [1ex] 
\bottomrule
\end{tabular}
\caption{Values for the $a$ parameter of the three models.} 
\label{tbl:sdparamsvalCB} 
\end{table}

\begin{figure}
    \centering
    \includegraphics[scale=0.25]{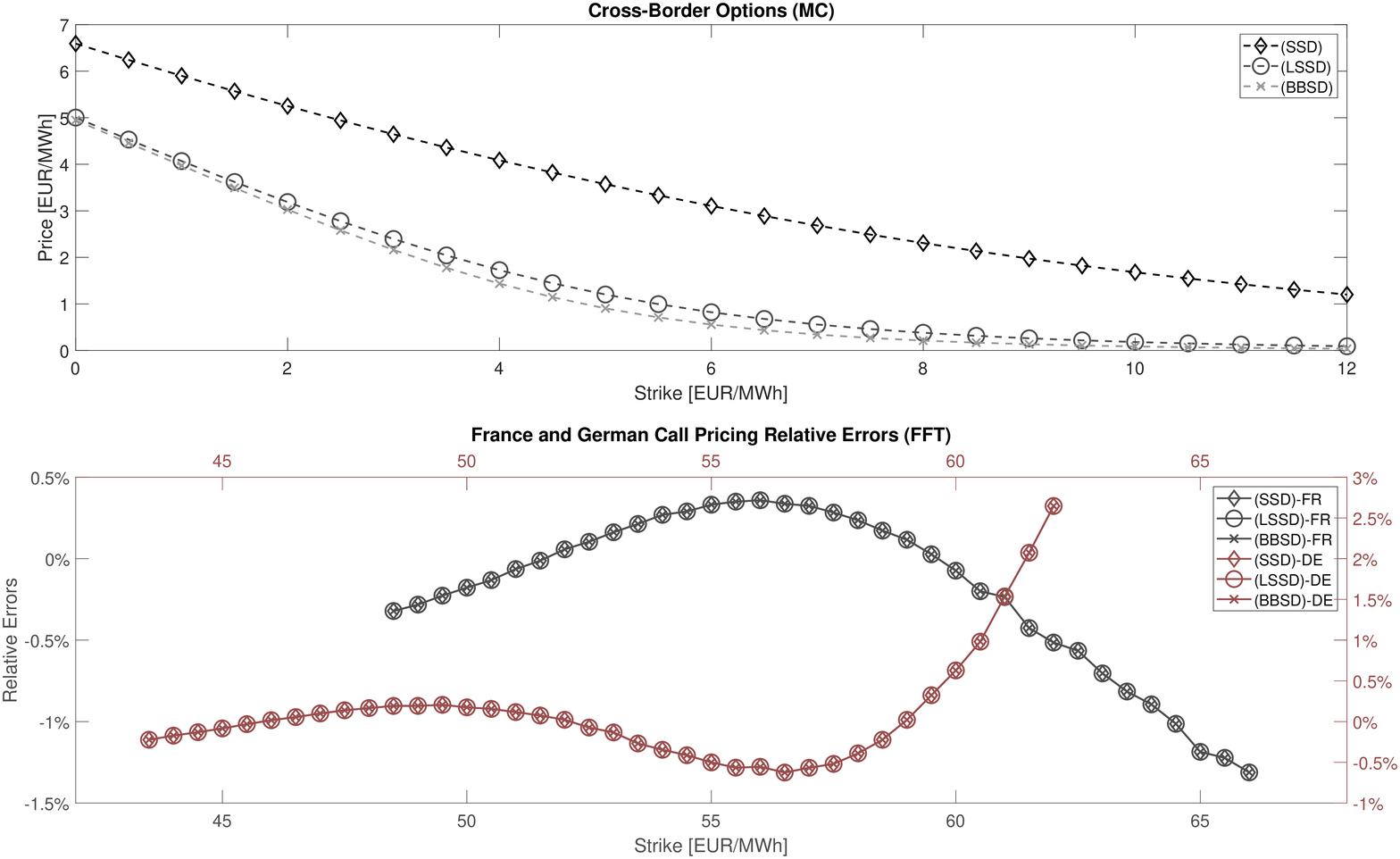}
    \caption{Percentage errors and Cross Border option prices.}
    \label{fig:CrossBorderOption}
\end{figure}

\begin{table}[ht]
{\small
\centering 
\begin{tabular}{c|ccc|ccc|ccc}
\toprule
\multicolumn{1}{c}{K} & \multicolumn{3}{c}{ \emph{SSD-NIG}} & \multicolumn{3}{c}{ \emph{LSSD-NIG}} & \multicolumn{3}{c}{ \emph{BBSD-NIG}} \\
\midrule 
 - & FFT & MC & $\Delta$ & FFT & MC & $\Delta$ & FFT & MC & $\Delta$ \\
\midrule
0.0 & 6.61 & 6.59 & (0.02) & 5.01 &5.00 & (0.01) & 4.95 &4.95 &   (0.00)  \\ 
0.5 & 6.26 &6.24 & (0.02) & 4.54 &4.53 & (0.01) & 4.46 &4.46 &   (0.00) \\ 
1.0 & 5.92 &5.90 & (0.02) & 4.08 &4.07 & (0.01) & 3.98 &3.97 &    (0.01) \\ 
1.5 & 5.59 &5.57 & (0.02) & 3.63 &3.62 & (0.01) & 3.50 &3.49 & (0.01)\\ 
2.0 & 5.27 &5.25 & (0.02) & 3.20 &3.19 & (0.01) & 3.03 &3.03 & (0.00)\\ 
2.5 & 4.96 &4.94 & (0.02) & 2.79 &2.78 & (0.01) & 2.58 &2.58 & (0.00)\\ 
3.0 & 4.67 &4.65 & (0.02) & 2.41 &2.40 & (0.01) & 2.16 &2.16 & (0.00)\\ 
3.5 & 4.38 &4.36 & (0.02) & 2.06 &2.04 & (0.02) & 1.78 &1.78 & (0.00)\\ 
4.0 & 4.11 &4.09 & (0.02) & 1.74 & 1.73 & (0.01) & 1.44 &1.44 & (0.00)\\ 
4.5 & 3.85 &3.82 & (0.03) & 1.46 &1.45 & (0.01) & 1.15 &1.15 & (0.00)\\ 
5.0 & 3.59 &3.57 & (0.02) & 1.22 &1.20 &  (0.02) & 0.91 &0.90 & (0.01)\\ 
5.5 & 3.36 &3.33 & (0.03) & 1.01 &1.00 &  (0.01) & 0.71 &0.71 & (0.00)\\ 
6.0 & 3.13 &3.11 & (0.02) & 0.83 &0.82 &  (0.01) & 0.56 &0.56 & (0.00)\\ 
6.5 & 2.91 &2.89 & (0.02) & 0.69 &0.68 &  (0.01)  & 0.44 &0.44 & (0.00)\\  
7.0 & 2.71 &2.68 & (0.03)& 0.57 &0.56 &  (0.01)  & 0.34 &0.34 & (0.00)\\ 
7.5 & 2.51 &2.49 & (0.02)& 0.47 &0.46 &  (0.01)  & 0.27 &0.27 & (0.00)\\ 
8.0 & 2.33 &2.31 & (0.02)& 0.39 &0.38 &  (0.01)  & 0.21 &0.21 & (0.00)\\ 
8.5 & 2.16 &2.14 & (0.02)& 0.32 &0.32 &  (0.00) & 0.17 &0.17 & (0.00)\\ 
9.0 & 2.00 &1.97 & (0.03)& 0.27 &0.26 &  (0.01) & 0.14 &0.13 & (0.01)\\ 
9.5 & 1.84 &1.82 & (0.02)& 0.22 &0.22 &  (0.00) & 0.11 &0.11 & (0.00)\\ 
10.0 & 1.70 &1.68 & (0.02)& 0.19 &0.18 &  (0.01)  & 0.09 &0.09 & (0.00)\\ 
10.5 & 1.57 &1.55 & (0.02)& 0.16 &0.15 &  (0.01)  & 0.07 &0.07 & (0.00)\\ 
11.0 & 1.44 &1.42 & (0.02)& 0.13 &0.13 &  (0.00) & 0.06 &0.06 & (0.00)\\ 
11.5 & 1.33 &1.31 &(0.02) & 0.11 &0.11 & (0.00) & 0.05 &0.05 & (0.00)\\ 
12.0 & 1.22 &1.20 & (0.02)& 0.09 &0.09 & (0.00) & 0.04 &0.04 & (0.00)\\ 
\bottomrule
\end{tabular}
\caption{Cross Border Option prices comparison between three models. Option prices are obtained using both FFT and MC methods. $\Delta$ is the difference between prices.} 
\label{tbl:CrossBorderoptionpricesFFTvsMC} 
}
\end{table}

%% file: Conclusion.tex
Using the concept of self-decomposable subordinators introduced by \citet{gardini2020}, we have shown how some recently proposed multivariated \Levy\ models can be easily extended to include what we called \textit{synaptic risk}. Based on this machinery, we build new bivariate versions of a Normal Inverse Gaussian process aiming at capturing stochastic delays. Their mathematical tractability were preserved and, moreover, we derived closed form expressions for their characteristic functions and linear correlation coefficients. These results were instrumental to apply calibration and derivative pricing methods based on Fourier techniques. 
\par Nevertheless, in many real applications, Monte Carlo simulations are required for complex derivative contracts pricing. Basing on some observations in \citet{Taufer2009} and \citet{cs20} about the strong mathematical connection between self-decomposable laws and \Levy-driven Ornstein-Uhlenbeck processes, we developed a new efficient algorithm to generate the \textit{a}-reminder of Inverse Gaussian law and hence to simulate the desired Normal Inverse Gaussian process with stochastic delays. The just mentioned algorithm is more efficient than the one proposed by \citet{Zhang2008}, because it is not based on acceptance-rejection methods: for this reason it can be adopted for real time simulations and pricing. 
\par Eventually, we applied these results to the modeling of energy markets: using the two-steps calibration technique proposed by \citet{SL2010}, all presented models have been calibrated on vanilla products and on historical quotations and, finally, commonly traded derivative contracts, such as \textit{Cross-Border} or \textit{Spark-Spread} options, have been efficiently priced using both Monte Carlo simulations and the Fourier method proposed by \citet{caldanafusai2016}.
\par In this article, we did not give a complete characterization of the \Levy\ process which can be built starting from the $a$-reminder of a self-decomposable law. For this reason it might be worth deeply investigating mathematical properties of such a process and those of the one obtained subordinating a standard Brownian Motion with it. 
\par It is a well known fact that Inverse Gaussian and Gamma laws are special cases of Generalized Inverse Gaussian laws which are self-decomposable, as was shown by \citet{Halgreen1979}. \citet{Zhang2011TransitionLS} gave a complete characterization of Ornstein-Uhlenbeck processes with Generalize Inverse Gaussian stationary laws: their numerical simulations, achieved by extending our new aforementioned approach, might be the object of a future research. 
\par Many exotic derivatives widely traded in energy markets, such as swing and storage contracts, require Least Squares Monte Carlo approach in order to be valued: time reversal simulations approach presented in \citet{PellegrinoSabino15} and \citet{Sabino20} might be adapted to simulate backward in time above mentioned processes leading to efficient pricing algorithms: therefore, this topic will be the subject of future inquires.

%% file: differentparametrizations.tex
\label{sec:IGlaws}
The characterization of the \pdf\ of an IG law is not unique. 
For example, \citet{CT2003} proposed a parameters setting in $\left(\mu,\lambda\right)$, that we denoted by $IG_{T}\left(\mu,\lambda\right)$ where $\mu>0$ is the mean and $\lambda>0$ is the shape parameter. Within this setting the \pdf\ of an Inverse Gaussian law is given by:

\begin{equation}
f_{Z}\left(x; \mu,\lambda\right) = \left(\frac{\lambda}{2\pi x^3}\right)^{1/2} \exp\left\{-\frac{\lambda\left(x-\mu\right)^2}{2\mu^2 x}\right\}
\label{eqn:IGdensity}
\end{equation}
and its \chf\ is:

\begin{equation}
\phi_{Z}\left(u\right) = \exp\left\{\frac{\lambda}{\mu}\left[1 - \sqrt{1 - \frac{2iu\mu^2}{\lambda}}\right]\right\}
\label{eqn:chfun1}
\end{equation}
Moreover let be $X\sim IG_{T}\left(\mu,\lambda\right)$ then we have that:

\begin{equation*}
\mathbb{E}\left[X\right] = \mu, \quad Var\left[X\right] = \frac{\mu^3}{\lambda}
\end{equation*}


The original parameter setting of a IG law proposed by \citet{BN97} is denoted with $IG_{B}\left(a,b\right)$, where $a$ can is the scale parameter and $b$ represents the shape of the distribution. Its probability density function is given by:

\begin{equation}
f_{Z}\left(x;a,b\right) = \frac{a}{\sqrt{2\pi}} \exp\left(ab\right)x^{-3/2}\exp\left(-\frac{1}{2}\left(a^2x^{-1} + b^2 x\right)\right)
\label{eqn:IGdensity0}
\end{equation}
and the \chf\ has the following form:
\begin{equation}
\phi_{Z}\left(u\right) = \exp\left\{-a\left(\sqrt{-2iu + b^2} - b\right)\right\}
\label{eqn:chfun2}
\end{equation}
If $X\sim IG_{B}\left(a,b\right)$ then we have that:

\begin{equation*}
\mathbb{E}\left[X\right] = \frac{a}{b}, \quad
Var\left[X\right] = \frac{a}{b^{3}}
\end{equation*}

Both parametrizations can be adopted and it is possible to switch from one the other by observing that: 
\begin{align}
\mu & = \frac{a}{b} \\
\lambda & = a^{2}
\label{eqn:equivalentrapp}
\end{align}

We report some very useful properties of the IG law.
\begin{itemize}
\item Let be $X\sim IG_{B}\left(a_{1},b\right)$ and $Y\sim IG_{B}\left(a_{2},b\right)$ and let $X$ and $Y$ be independent. Then:
\begin{equation*}
cX \sim IG_{B}\left(ca_{1},\frac{b}{c}\right),\quad X + Y  \sim IG_{B}\left(a_{1}+a_{2},b\right)
\end{equation*}
\item Let be $X\sim IG_{T}\left(\mu_{0}w_{1},\lambda_{0}w_{1}^{2}\right)$ and $Y\sim IG_{T}\left(\mu_{0}w_{2},\lambda_{0}w_{2}^{2}\right)$ and let $X$ and $Y$ be independent. Then:
\begin{equation*}
cX \sim IG_{T}\left(c\mu_{0}w_{1},c\lambda_{0}w_{1}\right), \quad X + Y  \sim IG_{T}\left(\mu_{0}\left(w_{1}+w_{2}\right),\lambda_{0}\left(w_{1}+w_{2}\right)^{2}\right)
\end{equation*}
\end{itemize}